\documentclass[onecolumn]{IEEEtran}
\usepackage{amsfonts,amsmath,amssymb}

\def\BibTeX{{\rm B\kern-.05em{\sc i\kern-.025em b}\kern-.08em
    T\kern-.1667em\lower.7ex\hbox{E}\kern-.125emX}}
\newtheorem{theorem}{Theorem}

\newtheorem{corollary}{Corollary}

\newtheorem{definition}{Definition}
\newtheorem{example}{Example}

\newtheorem{lemma}{Lemma}

\begin{document}

\title{On the Construction of Skew Quasi-Cyclic Codes$^{\dag }$}
\author{Taher Abualrub\thanks{%
T. Abualrub is with the Department of Mathematics and Statistics, American
University of Sharjah, Sharjah, UAE (e-mail: abualrub@aus.edu).}, Ali Ghrayeb%
\thanks{%
A. Ghrayeb is with the Department of Electrical and Computer Engineering,
Concordia University, Montreal, Quebec, Canada (e-mail:
aghrayeb@ece.concordia.ca).}, Nuh Aydin, and Irfan\thanks{%
N. Aydin is with the Department of Mathematics, Kenyon College, Gambier, OH,
USA (e-mail: aydinn@kenyon.edu). } Siap\thanks{%
I. Siap is with the Education Faculty, Adiyaman Univ., Adiyaman, Turkey
(e-mail: isiap@adiyaman.edu.tr).}}
\maketitle

\begin{abstract}%

In this paper we study a special type of quasi-cyclic (QC) codes
called skew QC codes. This set of codes is constructed using a
non-commutative ring called the skew polynomial rings $F[x;\theta
]$. After a brief description of the skew polynomial ring
$F[x;\theta ]$ it is shown that skew QC codes are left submodules
of the ring $R_{s}^{l}=\left( F[x;\theta ]/(x^{s}-1\right) )^{l}.$
The notions of generator and parity-check polynomials are given.
We also introduce the notion of similar polynomials in the ring
$F[x;\theta ]$ and show that parity-check polynomials for skew QC
codes are unique up to similarity. Our search results lead to the
construction of several new codes with Hamming distances exceeding
the Hamming distances of the previously best known linear codes
with comparable parameters.

\end{abstract}%

\section{Introduction}

\footnotetext{$^{\dag }$The work of T. Abualrub and A. Ghrayeb was supported
in part by the Natural Sciences and Engineering Research Council of Canada
(NSERC) while the first author was on sabbatical in the Department of
Electrical and Computer Engineering, Concordia University, Montreal, Quebec,
Canada.} Since the introduction of Shannon theory in 1948, coding theorists
have been trying to design powerful codes that approach the Shannon capacity
with reasonable complexity. Initially the focus was on designing codes that
posses large minimum distances, resulting in several classes of linear block
codes and convolutional codes. With the invention of turbo codes in 1993,
where more focus has been given to reducing the multiplicity of the minimum
distance rather than increasing the minimum distance itself, concatenated
codes that perform within a few tenths of a dB from capacity have been
designed and incorporated in a number of communications applications
standards. However, since these codes have been developed for additive white
Gaussian noise (AWGN) channels, they may not be the best choice for most
wireless communications applications in which the channel normally suffers
from severe fading. Such applications include cellular networks, wireless
local area networks and wireless sensor networks, to name a few. To this
end, a relatively new class of codes termed \textit{space-time coding} has
been introduced, which includes space-time trellis codes \cite{tarokh98} and
space-time block codes \cite{alamouti98}, \cite{tarokh99}. Such codes are
suitable for multiple-input multiple-output (MIMO) systems where the
transmitter and/or receiver are equipped with multiple antennas. It has been
shown that a MIMO\ system with $N_{t}$ transmit and $N_{r}$ receive antennas
achieves a spatial diversity of $N_{t}N_{r}$ \cite{tarokh98}$-$\cite%
{tarokh99}$.$ To achieve better performance, one may need to combine error
correcting coding with space-time coding since the former coding method
introduces temporal diversity. For example, in a coded MIMO system, the
maximum diversity that can be achieved (over block faded channels with
proper interleaving) can be as high as $N_{t}N_{r}d_{\min }^{H},$ where $%
d_{\min }^{H}$ denotes the minimum Hamming distance of the error correcting
code employed \cite{duman08}. This suggests that having a code with a high
minimum Hamming distance is essential since it offers significant
performance improvements.

A significant portion of the work on error correcting codes for over the
last sixty years has been on the construction of different types of codes
defined over commutative rings. At the beginning, most of the research on
error correcting codes was concentrated on codes over finite fields. More
recently, it has been shown by many researchers (e.g., \cite{nuh},\cite%
{cyclicZ4},\cite{newcode},\cite{Hamm}) that codes over rings are a very
important class and many types of codes with good parameters can be
constructed over rings. We believe that another important direction to
consider is the construction of codes using non-commutative rings. Research
on this topic is very recent and interesting. Boucher, et. al generalized in
\cite{Boucher07}, \cite{Boucher07pri} the notion of cyclic codes by using
generator polynomials in a non-commutative polynomial ring called skew
polynomial ring. They gave examples of skew cyclic codes with Hamming
distances larger than previously best known linear codes of the same length
and dimension \cite{Boucher07}.

Quasi-cyclic (QC) codes of index $l$ over a finite field $F$ are linear
codes where the cyclic shift of any codeword by $l$ positions is another
codeword. QC codes of index $l=1$ are well known cyclic codes. QC codes have
been shown to be a very important class of linear codes \cite{chen}, \cite%
{kasami}, \cite{grobner2}, \cite{sonqc}, \cite{seguin}, and \cite{koshy}.
Many of the best known and optimal linear codes that have been constructed
so far are examples of QC codes (e.g., \cite{chen}, \cite{nine},\cite{best},%
\cite{optimal}, and \cite{ieee}.)

In this paper we study the construction of skew QC codes. This work has been
motivated by the fact that the class of skew QC codes is much larger than
the class of QC codes, suggesting that better codes may be found in this
class. Indeed, we have found many examples of skew QC codes that meet or
exceed the parameters of best known linear codes. In particular we are
interested in the study of 1-generator skew QC codes and their properties.
We show that this class of codes share many properties of QC codes.

The rest of the paper is organized as follows. Section II includes a brief
description of the skew polynomial ring $F[x;\theta ].$ In Section III, we
discuss the structure of skew QC codes where we show that this type of codes
is a left submodule of $R_{s}^{l}=\left( F[x;\theta ]/(x^{s}-1\right) )^{l}.$
We also discuss the dimension and the parity check polynomial for these
codes. In Section IV we introduce the notion of similar polynomials. We will
show that the parity-check polynomial of a skew QC code is unique up to
similarity. Section V includes our search results. As a result of a search
in the class of skew QC codes over $GF(4)$, we obtain seven new linear
quaternary codes with Hamming distances greater than previously best known
linear codes with the given parameters. These new codes have the parameters $%
\left[ 48,12,24\right] $, $\left[ 72,21,29\right] ,$ $\left[ 48,16,20\right]$%
, $[96,16,49]$, $[100,20,47]$, $[140,20,72]$, and $[110,22,51]$.
We also construct a large number of skew QC codes with Hamming
distances equal to the Hamming distances of the best known linear
codes with the given parameters. Section VI concludes the paper.

\bigskip

\section{The Skew Polynomial Ring $F[x;\protect\theta ]$}

Let $F$ be a finite field of characteristic $p$. Let $\theta $ be an
automorphism of \ $F$ with $\left\vert \left\langle \theta \right\rangle
\right\vert =m.$ Let $K$ be the subfield of $F$ fixed under $\left\langle
\theta \right\rangle .$ Then, $\left[ F:K\right] =m$ and $K=GF(p^{t}),$ $%
F=GF(q)$ where $q=p^{tm}.$ Moreover since $K$ is fixed under $\theta $ then
we have $\theta \left( a\right) =a^{p^{t}}$ for all $a\in F.$

\begin{example}
\label{Example1}Consider the finite field $GF(4)=\left\{ 0,1,a ,a^2\right\} $
where $a ^2+a +1=0.$ Define an automorphism
\begin{eqnarray*}
\theta &:&GF(4)\rightarrow GF(4)\text{ by} \\
\theta (z) &=&z^2.
\end{eqnarray*}

Then $\theta (0)=0,~\theta (1)=1,~\theta (a)=a ^2$ and $\theta (a ^2)$ $=a .$
Hence the fixed field $K$ is just the binary field $GF(2).$
\end{example}

\begin{definition}
Following the above notation, define the skew polynomial set $F[x;\theta ]$
to be
\begin{equation*}
F[x;\theta ]=\left\{
\begin{array}{c}
f(x)=a_{0}+a_{1}x+a_{2}x^{2}+\cdots +a_{n}x^{n}|~\text{\ where } \\
a_{i}\in \text{ }F\text{ for all }i=0,1,\ldots ,n%
\end{array}%
\right\}
\end{equation*}
\end{definition}

where addition of these polynomials is defined in the usual way while
multiplication is defined using the distributive law and the rule
\begin{equation*}
(ax^{i})(bx^{j})=a\theta ^{i}(b)x^{i+j}.
\end{equation*}

\begin{example}
Using the same automorphism from Example\ref{Example1}, we get%
\begin{eqnarray*}
(a x)(a ^2x) &=&a \theta (a ^2)x^{2} \\
&=&a\cdot a x^2=a ^{2}x^{2}.
\end{eqnarray*}%
On the other hand we have,%
\begin{eqnarray*}
(a ^{2}x)(a x) &=&a ^{2}\theta (a )x^{2} \\
&=&a ^{2}(a ^{2})x^{2}=a x^{2}.
\end{eqnarray*}%
This shows that $(a x)(a ^{2}x)\neq $ $(a ^{2}x)(a x).$
\end{example}

\begin{theorem}
\cite{Mcdonald74} The set $F[x;\theta ]$ with respect to addition and
multiplication defined above forms a non-commutative ring called the skew
polynomial ring.
\end{theorem}

The following facts are straightforward for the ring $F[x;\theta ]:$

\begin{enumerate}
\item It has no nonzero zero-divisors.

\item The units of $F[x;\theta ]$ are the units of $F.$

\item $\deg (f+g)\leq $\textbf{$\max \{\deg (f),\text{ }\deg (g)\}$}

\item $\deg (fg)=\deg (f)+\deg (g).$
\end{enumerate}

The skew polynomial ring $F[x;\theta ]$ was introduced by Ore \cite{Ore33}
in 1933, and a complete treatment of this ring can be found in \cite%
{Jacobson43} and in \cite{Mcdonald74}.

\begin{theorem}
\cite{Mcdonald74} (The Right Division Algorithm) For any polynomials $f$ and
$g$ in $F[x;\theta ]$ $\ $with $f\neq 0$ there exist unique polynomials $q$
and $r$ such that%
\begin{equation*}
g=qf+r\text{ where }\deg (r)<\deg (f).
\end{equation*}
\end{theorem}

The above result is called division on the right by $f.$ A similar result
can be proved regarding division on left by $f.$

Applying the division algorithm above one can easily prove the following
Theorem.

\begin{theorem}
\label{Theorem6}\cite{Jacobson43} $F[x;\theta ]$ is a non-commutative
principal left (right) ideal ring. Moreover any two sided ideal must be
generated by%
\begin{equation*}
f(x)=\left( a_{0}+a_{1}x^{m}+a_{2}x^{2m}+\ldots +a_{r}x^{rm}\right) x^{t},
\end{equation*}%
where $\left\vert \left\langle \theta \right\rangle \right\vert =m.$
\end{theorem}

\begin{corollary}
\label{cor2}Let $\theta $ be an automorphism of \ $F$ with $\left\vert
\left\langle \theta \right\rangle \right\vert =m.$ Then $(x^{s}-1)$ is a two
sided ideal in $F[x;\theta ]$ iff $m|s$.
\end{corollary}

\begin{lemma}
\label{Lemma7} $(x^{s}-1)\in Z\left( F[x;\theta ]\right) $ for $m|s$, where $%
Z\left( F[x;\theta ]\right) $ is the center of $F[x;\theta ].$
\end{lemma}

\begin{proof}
Let $f(x)=a_{0}+a_{1}x+\cdots +a_{r}x^{r}\in F[x;\theta ]$. Since $m|s,$
then $\theta ^{s}(a)=a$ for any $a\in F.$ Hence,%
\begin{eqnarray*}
\left( x^{s}-1\right) f(x) &=&\left( x^{s}-1\right) \left(
a_{0}+a_{1}x+\cdots +a_{r}x^{r}\right) \\
&=&a_{0}x^{s}+a_{1}x^{s+1}+\cdots +a_{r}x^{s+r} \\
&&-a_{0}-a_{1}x+\cdots -a_{r}x^{r} \\
&=&\left( a_{0}+a_{1}x+\cdots +a_{r}x^{r}\right) \left( x^{s}-1\right) \\
&=&f(x)\left( x^{s}-1\right) .
\end{eqnarray*}
\end{proof}

\begin{lemma}
\label{Lemma8}\cite{Boucher07pri} If $g\cdot h\in Z\left( F[x;\theta
]\right) $ then $g\cdot h=h\cdot g$
\end{lemma}

From Theorem \ref{Theorem6}, Lemma \ref{Lemma7}, and Lemma \ref{Lemma8}, we
may conclude that the factors of $x^{s}-1$ commute. Thus if $f$ is a left
divisor then it is a right divisor as well. This fact will help in reducing
the complexity of factoring $x^{s}-1$ in $F[x;\theta ].$ From now on we will
say divisors or factors of $x^{s}-1$ without specifying left or right.

\begin{definition}
A polynomial $f$ is called a left multiple of a polynomial $d$ (in this case
$d$ will be called a right divisor of $f$) if there exits a polynomial $g$
such that%
\begin{equation*}
f=gd.
\end{equation*}
\end{definition}

\begin{definition}
A monic polynomial $d$ is called the greatest common right divisor (gcrd) of
$f$ and $g$ if

\begin{enumerate}
\item $d$ is a right divisor of $f$ and $g,$ and

\item If $e$ is another right divisor of $f$ and $g$ then $d=ke$ for some
polynomial $k.$
\end{enumerate}
\end{definition}

The greatest common left divisor (gcld) of $a$ and $b$ is a monic polynomial
defined in a similar way. Similarly we define the least common right
multiple of $a$ and $b$ lcrm$[a,b]$ and the least common left multiple of $a$
and $b,$ lclm$[a,b]$

\begin{theorem}
\label{Theorem3}\cite{Ore33} gcrd, gcld, lcrm, and lclm can be calculated
using the left and right division algorithms.
\end{theorem}

\section{Skew Quasi Cyclic Codes}

\begin{definition}
Let $F$ be a finite field of characteristic $p$ with $q=p^{mt}$ elements,
and let $\theta $ be an automorphism of \ $F$ with $\left\vert \left\langle
\theta \right\rangle \right\vert =m.$ A subset $C$ of $F^{n}$ is called a
skew quasi-cyclic code of length $n$ where $n=sl,$ $m|s,$ and index $l$ if

\begin{enumerate}
\item $C$ is a subspace of $F^{n}.$

\item If
\begin{equation*}
c=\left(
\begin{array}{c}
c_{0,0},c_{0,1},\ldots ,c_{0,l-1},c_{1,0},c_{1,1},\ldots ,c_{1,l-1},\ldots ,
\\
c_{s-1,0},c_{s-1,1},\ldots ,c_{s-1,l-1}%
\end{array}
\right) \in C
\end{equation*}
then
\begin{equation*}
T_{\theta,s,l}(c)=\left(
\begin{array}{c}
\theta (c_{s-1,0}),\theta (c_{s-1,1}),\ldots ,\theta (c_{s-1,l-1}),\theta
(c_{0,0}), \\
\ldots ,\theta (c_{0,l-1}),\ldots ,\theta (c_{s-2,0}),\ldots ,\theta
(c_{s-2,l-1})%
\end{array}
\right) \in C.
\end{equation*}
\end{enumerate}
\end{definition}

The map $T_{\theta,s,l}$ will be referred to as skew cyclic shift
operator. Thus skew QC codes are linear codes that are closed
under skew cyclic shift. If $\theta $ is the identity map, then
skew QC codes are just the standard QC codes defined over $F.$

In \cite{Boucher07}, Boucher, etc. studied skew cyclic codes over $F.$ They
showed that a code $C$ is a skew cyclic code if and only if $C$ is a left
ideal generated by $g(x)$ where $g(x)$ is a right divisor of $x^{n}-1.$

Recall from Corollary \ref{cor2} that $x^{s}-1$ is a two sided ideal iff $%
m|s.$ Because of this, we will always assume that $C$ is a skew quasi-cyclic
code of length $n$ where $n=sl,$ $m|s,$ and index $l.$

In this paper we focus on skew QC codes over the finite field $F=GF(4)$ even
though most results can be generalized to any finite field.

The ring $R_{s}^{l}=\left( F[x;\theta ]/(x^{s}-1\right) )^{l}$ is a left $%
R_{s}=F[x;\theta ]/(x^{s}-1)$ module where we define multiplication from
left by
\begin{eqnarray*}
&&f(x)\left( g_{1}(x),g_{2}(x),\ldots ,g_{l}(x)\right) \\
&=&\left( f(x)g_{1}(x),f(x)g_{2}(x),\ldots ,f(x)g_{l}(x)\right) .
\end{eqnarray*}

Let $c=\left(
\begin{array}{c}
c_{0,0},c_{0,1},\ldots ,c_{0,s-1},c_{1,0},c_{1,1},\ldots ,c_{1,l-1}, \\
\ldots ,c_{l-1,0},c_{l-1,1},\ldots ,c_{l-1,s-1}%
\end{array}%
\right)$ be an element in $F^{sl}$. Define a map $\phi :F^{sl}\rightarrow
R_{s}^{l}$ by
\begin{eqnarray*}
\phi(c)=(c_{0}(x),c_{1}(x),\ldots, c_{l-1}(x))
\end{eqnarray*}%
where
\begin{equation*}
c_{j}(x)=\sum_{i=0}^{s-1}c_{i,j}x^{i}\in F[x;\theta ]/(x^{s}-1)\text{ for }%
j=0,1,\ldots ,l-1.
\end{equation*}%
The map $\phi $ gives a one to one correspondence between the ring $F^{sl}$
and the ring $R_{s}^{l}.$ It is also a vector space isomorphism between $%
F^{sl}$ and $R_{s}^{l}$, when considered as vector spaces over $F$.

\begin{theorem}
A subset $C$ of $F^{n}$ is a skew QC code of length $n=sl$ and index $l$ if
and only if $\phi \left( C\right) $ is a left submodule of the ring $%
R_{s}^{l}.$
\end{theorem}

\begin{proof}
Let $C$ be a skew QC code of index $l$ over $F.$ We claim that $\phi (C)$
forms a submodule of $R_{s}^{l}$ where $\phi $ is the map defined above.
Clearly, $\phi (C)$ is closed under addition and scalar multiplication (by
elements of $F$). Let
\begin{equation*}
\phi (c)=(c_{0}(x),c_{1}(x),\ldots ,c_{l-1}(x))\in \phi (C).
\end{equation*}
for

\begin{equation*}
c=\left(
\begin{array}{c}
c_{0,0},c_{0,1},\ldots ,c_{0,s-1},c_{1,0},c_{1,1},\ldots , \\
c_{1,s-1},\ldots ,c_{l-1,0},c_{l-1,1},\ldots ,c_{l-1,s-1}%
\end{array}%
\right) \in C
\end{equation*}%
.

Then
\begin{eqnarray*}
x\phi (c) &=&(xc_{0}(x),xc_{1}(x),\ldots ,xc_{l-1}(x)) \\
&=&\left(
\begin{array}{c}
\theta \left( c_{s-1,0}\right) +\theta \left( c_{0,0}\right) x+\cdots + \\
\theta \left( c_{s-2,0}\right) x^{s-1},\theta \left( c_{s-1,1}\right)
+\theta \left( c_{0,1}\right) x+\cdots \\
+\theta \left( c_{s-1,1}\right) x^{s-1},\ldots ,\theta \left(
c_{s-1,l-1}\right) + \\
\theta \left( c_{0,l-1}\right) x+\cdots +\theta \left( c_{s-1,l-1}\right)
x^{s-1}%
\end{array}%
\right) \\
&=&\phi \left(
\begin{array}{c}
\theta \left( c_{s-1,0}\right) ,\theta \left( c_{s-1,1}\right) ,\ldots , \\
\theta \left( c_{s-1,l-1}\right) ,\theta \left( c_{0,0}\right) ,\theta
\left( c_{0,1}\right) ,\ldots , \\
\theta \left( c_{0,l-1}\right) ,\ldots ,\theta \left( c_{s-2,0}\right) , \\
\theta \left( c_{s-2,1}\right) ,\ldots ,\theta \left( c_{s-2,l-1}\right)%
\end{array}%
\right) \in \phi (C).
\end{eqnarray*}

Then, by linearity it follows that $p(x)\phi(c)\in \phi (C)$ for any $%
p(x)\in R_s$. Hence $\phi (C)$ is a left submodule of $R_{s}^{l}.$

Conversely, suppose $D$ is an $R_{s}$ left submodule of $R_{s}^{l}.$ Let $%
C=\phi^{-1}(D)=\{c\in F^n: \phi(c)\in D \}$. We claim that $C$ is
a skew QC code over $F$. Since $\phi$ is a vector space
isomorphism, $C$ is a linear code of length $n$ over $F$. To show
that $C$ is closed under skew cyclic shift, let $c =\left(
c_{0,0},c_{0,1},\ldots ,c_{0,s-1},c_{1,0},c_{1,1}, \newline \ldots
,c_{1,s-1},\ldots ,c_{l-1,0},c_{l-1,1},\ldots ,c_{l-1,s-1}
\right)\in C$. Then, $\phi(c)= \left( g_{0}(x),g_{1}(x),\ldots,
g_{l-1}(x)\right) \in D$, where
$g_{j}(x)=\sum_{i=0}^{s-1}c_{i,j}x^{i} \text{ for } j=0,1,\ldots
,l-1$. From the above discussion, it is easy to see that
$\phi\left( T_{\theta,s,l}(c) \right)=x\left(
g_{0}(x),g_{1}(x),\ldots g_{l-1}(x)\right)=\left(
xg_{0}(x),xg_{1}(x),\ldots, xg_{l-1}(x)\right) \in D$. Hence
$T_{\theta,s,l}(c)\in C$. Therefore $C$ is a skew quasi-cyclic
code $C$.
\end{proof}

From now on we concentrate on 1-generator skew QC codes that are cyclic left
submodules of $R_{s}^{l}.$ i.e. $\ $ any skew QC code $C$ that has the form
\begin{equation*}
C=\left\{ f(x)\left( g_{1}(x),g_{2}(x),\ldots ,g_{l}(x)\right) :\text{ }%
f(x)\in R_{s}\text{ }\right\} .
\end{equation*}%
Sometimes we denote this by
\begin{equation*}
C=\left\{
\begin{array}{c}
f(x)G(x):\text{ }f(x)\in R_{s}\text{ \ and } \\
G(x)=\left( g_{1}(x),g_{2}(x),\ldots ,g_{l}(x)\right)%
\end{array}%
\right\} .
\end{equation*}

\begin{theorem}
Let $C$ be a one generator skew QC code of length $n=sl$ and index $l.$ Then
$C$ is generated by an element of the form
\begin{equation*}
\left( p_{1}(x)g_{1}(x),p_{2}(x)g_{2}(x),\ldots ,p_{l}(x)g_{l}(x)\right)
\end{equation*}%
where $g_{i}(x)$ is a divisor of $\left( x^{s}-1\right) .$
\end{theorem}

\begin{proof}
Let $C$ be a 1-generator skew QC code generated by $\left(
f_{1},~f_{2},\ldots ,~f_{l}\right) $. For all $1\leq i\leq l$ define the
following map%
\begin{eqnarray*}
\Pi _{i} &:&C\rightarrow R_{s}\text{ by} \\
\Pi _{i}\left( \left( kf_{1},kf_{2},\ldots ,kf_{l}\right) \right) &=&kf_{i}.
\end{eqnarray*}%
The function $\Pi _{i}$ is a \ module homomorphism. It is clear that the
image of $\Pi _{i}$ is a left ideal and thus is a skew cyclic code in $%
R_{s}. $ Therefore, $kf_{i}\in \Pi _{i}(C)=$ $\left( g_{i}\right) $ for all $%
i=1,2,\ldots ,l.$ Hence,%
\begin{equation*}
C=\left( p_{1}(x)g_{1}(x),p_{2}(x)g_{2}(x),\ldots ,p_{l}(x)g_{l}(x)\right) ,
\end{equation*}%
where $g_{i}(x)$ is a divisor of $\left( x^{s}-1\right) .$
\end{proof}

\begin{definition}
Let
\begin{equation*}
C=\left( p_{1}(x)g_{1}(x),p_{2}(x)g_{2}(x),\ldots ,p_{l}(x)g_{l}(x)\right)
\end{equation*}%
be a skew QC code of length $n=sl$ and index $l.$ The unique monic polynomial%
\begin{equation*}
g(x)=gcld\left(
\begin{array}{c}
p_{1}(x)g_{1}(x),p_{2}(x)g_{2}(x), \\
\ldots ,p_{l}(x)g_{l}(x),x^{s}-1%
\end{array}%
\right)
\end{equation*}%
is called the generator polynomial of $C.$
\end{definition}

\begin{definition}
The monic polynomial $h(x)$ of minimal degree such that%
\begin{equation*}
h(x)\left(
\begin{array}{c}
p_{1}(x)g_{1}(x),p_{2}(x)g_{2}(x), \\
\ldots ,p_{l}(x)g_{l}(x)%
\end{array}%
\right) =\left( 0,0,\ldots ,0\right)
\end{equation*}%
is called the parity-check polynomial of $C$
\end{definition}

\begin{theorem}
\label{Lemma16}Suppose $d(x)=gcrd(f,g),$ then there are
polynomials polynomials $a(x),$ and $b(x)$ such that
\begin{equation*}
a(x)f(x)+b(x)g(x)=d(x).
\end{equation*}
\end{theorem}

\begin{proof}
The proof is similar to the case of $\gcd (a,b)$ when the ring is
commutative. Suppose $d(x)=gcrd(f,g).$ Consider the left ideal generated by $%
(f(x),g(x)).$ Since $F_{q}[x;\theta ]$ is a principal left ideal ring, there
exists a polynomial $h(x)$ such that\ $(f(x),g(x))=(h(x)).$ Hence $%
f(x)=r_{1}(x)h(x)$ and $g(x)=r_{2}(x)h(x).$ But $d(x)=gcrd(f,g)$ implies
that $d(x)=k(x)h(x)$ and $\left( d(x)\right) \subseteq (h(x)).$ Since $d(x)=$
$gcrd(f,g)$ then $f(x)$ and $g(x)\in $left ideal $\left( d(x)\right) .$
Hence $(h(x))\subseteq (d(x)$ and we have $(d(x))=(f(x),g(x))=(h(x)).$
Therefore there are polynomials $a(x),$ and $b(x)$ such that%
\begin{equation*}
a(x)f(x)+b(x)g(x)=d(x).
\end{equation*}
\end{proof}

\begin{corollary}
\label{cor1}Suppose\ $d(x)=gcld(f,g),$ then there are two polynomials $a(x),$
and $b(x)$ such that%
\begin{equation*}
f(x)a(x)+g(x)b(x)=d(x).
\end{equation*}
\end{corollary}

\begin{lemma}
\label{Lemma17} Let $g(x)$ and $h(x)$ be the generator and the parity-check
polynomials of a skew QC code $C.$ Then
\begin{equation*}
x^{s}-1=h(x)g(x)=g(x)h(x).
\end{equation*}
\end{lemma}

\begin{proof}
Since $g(x)=gcld\left(
\begin{array}{c}
p_{1}(x)g_{1}(x),p_{2}(x)g_{2}(x), \\
\ldots ,p_{l}(x)g_{l}(x),x^{s}-1%
\end{array}%
\right) ,$ then $x^{s}-1=g(x)k(x)$ for some polynomial $k(x).$ Note that by
Lemma \ref{Lemma8} we have $x^{s}-1=g(x)k(x)=k(x)g(x).$ Note also that $%
p_{i}(x)g_{i}(x)=g(x)\alpha _{i}(x)$ for all $i=1,\ldots ,l.$ Hence we have%
\begin{eqnarray*}
k(x)\left(
\begin{array}{c}
p_{1}(x)g_{1}(x),p_{2}(x)g_{2}(x), \\
\ldots ,p_{l}(x)g_{l}(x)%
\end{array}%
\right) &=& \\
k(x)\left(
\begin{array}{c}
g(x)\alpha _{1}(x),g(x)\alpha _{2}(x), \\
\ldots ,\alpha _{l}(x)g(x)\alpha _{l}(x)%
\end{array}%
\right) &=&\left( 0,0,\ldots ,0\right).
\end{eqnarray*}%
Hence $k(x)=q(x)h(x)$ and $\deg k(x)\geq \deg h(x).$ Now by Corollary \ref%
{cor1}, there are polynomials $a_{i}(x)$ such that%
\begin{eqnarray*}
&&p_{1}(x)g_{1}(x)a_{1}(x)+p_{2}(x)g_{2}(x)a_{2}(x)+\ldots \\
&&+(x^{s}-1)a_{l+1}(x) \\
&=&g(x).
\end{eqnarray*}%
Hence,%
\begin{eqnarray*}
&&h(x)p_{1}(x)g_{1}(x)a_{1}(x)+h(x)p_{2}(x)g_{2}(x)a_{2}(x) \\
&&+\ldots +h(x)(x^{s}-1)a_{l+1}(x) \\
&=&h(x)g(x) \\
0 &=&h(x)g(x).
\end{eqnarray*}%
This implies that $\deg h(x)\geq \deg \dfrac{x^{s}-1}{g(x)}=\deg k(x).$
Therefore $k(x)=h(x).$
\end{proof}

\begin{definition}
Let $C=\langle G(x)\rangle$ be a skew QC code. The annihilator of $C$ is the
set%
\begin{equation*}
I=\left\{ r(x):r(x)F(x)=0\text{ for all }F(x)\in C\right\} .
\end{equation*}%
It is clear that $I$ is a left ideal in $R_{s}.$
\end{definition}

\begin{lemma}
\label{Lemma 19} Let $C=G(x)$ be a skew quasi-cyclic code with annihilator $%
I.$ Then $I=\left( h(x)\right) $ and
\begin{eqnarray*}
C &\cong &R_{s}/I,\text{ and} \\
\dim C &=&\deg h(x).
\end{eqnarray*}
\end{lemma}

\begin{proof}
Define the map
\begin{eqnarray*}
\Psi &:&R_{s}\rightarrow C\text{ by} \\
\Psi (r(x) &=&r(x)G(x)
\end{eqnarray*}%
$\Psi $ is an onto module homomorphism with $\ker \Psi =I=\left( h(x)\right)
.$ Therefore $C\cong R_{s}/\left( h(x)\right) $ and hence $\dim C=\deg h(x).$
\end{proof}

\section{Similar Polynomials in $F[x;\protect\theta ].$}

In the case of QC codes in the ring $F[x],$ we know that the parity-check
polynomials are unique up to a unit. In the case of skew QC codes things are
not as straightforward as in the case of QC codes. To study the parity-check
polynomials we need to introduce the notion of similar polynomials in the
ring $F[x;\theta ].$ Our main result is to show that two polynomials $h_{1}$
and $h_{2}$ are parity-check polynomials for a code $C$ if and only if $%
h_{1} $ and $h_{2}$ are similar polynomials.

\begin{definition}
Two elements $a$ and $b$ in a ring $R$ are called right similar if there is
a $u\in R$ such that%
\begin{eqnarray*}
gcld(u,b) &=&1\text{ and } \\
ua &=&lcrm[u,b].
\end{eqnarray*}
\end{definition}

Left similar elements can be defined similarly.

\begin{example}
Let $F$ be any field of characteristic $p.$ We will show when two linear
polynomials $p_{1}(x)=x-\alpha $ and $p_{2}=x-\beta $ are right similar.

Let $u=c\in F,$ then $gcld\left( u,p_{2}\right) =1$ and $cc^{-1}p_{2}$ is a
right multiple of $u$ and $p_{2}.$ Hence, $[u,p_{2}]=cp_{1}$ iff%
\begin{equation*}
cc^{-1}p_{2}=cp_{1}\gamma =c(x-\alpha )\gamma \text{ for some }\gamma \in F.
\end{equation*}%
\newline
Hence,%
\begin{equation*}
x-\beta =c\theta (\gamma )-ca\gamma =c\gamma ^{p^{t}}-ca\gamma .
\end{equation*}%
This implies that%
\begin{eqnarray*}
1 &=&c\gamma ^{p^{t}}\text{ and\ }\beta =ca\gamma .\text{ This implies} \\
\beta &=&a\gamma ^{1-p^{t}}\text{ or }\alpha \beta ^{-1}=\gamma
^{p^{t}-1}\in F
\end{eqnarray*}%
Therefore, $x-\alpha $ is similar to $x-\beta $ iff $\alpha \beta
^{-1}=\gamma ^{p^{t}-1}\in F.$
\end{example}

If we consider the field $GF(2^{2})$ and the frobenoius automorphism we can
conclude that the polynomials $p_{1}(x)=x-1,$ $p_{2}(x)=x-\alpha $ and $%
p_{3}(x)=x-\alpha ^{2}$ are all right similar.

\begin{theorem}
If $a$ and $b$ are right similar then they are left similar.
\end{theorem}

\begin{proof}
Suppose there is a $u\in R$ such that
\begin{eqnarray*}
gcld(u,b) &=&1\text{ and } \\
ua &=&lcrm[u,b].
\end{eqnarray*}%
Let%
\begin{equation*}
m=ua=lcrm[u,b].
\end{equation*}%
Then $m=ua=bc$ for some $c.$ This shows that
\begin{equation*}
lclm[c,a]=m.
\end{equation*}%
Now suppose $gcrd(c,a)=d,$ then%
\begin{equation*}
c=\alpha _{1}d\text{ and }a=\alpha _{2}d.
\end{equation*}%
Hence,%
\begin{equation*}
m=ua=u\alpha _{2}d=bc=b\alpha _{1}d.
\end{equation*}%
This implies that
\begin{equation*}
lcrm[u,b]=u\alpha _{2}=b\alpha _{1}\neq m.
\end{equation*}%
A contradiction. Hence $gcrd(c,a)=1.$ This shows that $a$ and $b$ are left
similar.
\end{proof}

From now on we if $a$ and $b$ are right similar we will say that they are
similar. In the case that the ring is commutative then two elements are
similar iff they differ by a unit.

\begin{theorem}
Let $h_{1}(x)$ be a parity-check polynomial of a skew QC code $C_{1},$ and
let $h_{2}(x)$ be a parity-check polynomial of a skew QC code $C_{2}$ then $%
C_{1}=C_{2}$ iff $h_{1}(x)\,$\ is similar to $h_{2}(x).$
\end{theorem}

\begin{proof}
Suppose $C_{1}=C_{2}.$ Then $R_{s}/\left( h_{1}(x)\right) \cong R_{s}/\left(
h_{2}(x)\right) .$ Let%
\begin{equation*}
\Phi :R_{s}/\left( h_{1}(x)\right) \rightarrow R_{s}/\left( h_{2}(x)\right)
\end{equation*}%
be such a module isomorphism. Suppose $\Phi (1+\left( h_{1}(x)\right)
)=a+(h_{2}(x))$ then
\begin{equation}
\Phi (r+\left( h_{1}(x)\right) )=ra+(h_{2}(x))\text{ for any }r\in R_{s}.
\label{EQ1}
\end{equation}%
In particular we have
\begin{equation*}
\Phi (h_{1}+\left( h_{1}(x)\right) )=h_{1}a+(h_{2}(x)).
\end{equation*}%
Since $\Phi $ is a module isomorphism we must have%
\begin{equation*}
\Phi (h_{1}+(h_{1}(x)))=h_{2}(x)=0.
\end{equation*}%
This implies that $h_{1}a\in (h_{2}(x))$ and hence $h_{1}a=r_{2}h_{2}=m.$

Since $\Phi $ is surjective then there is $c\in R$ such that
\begin{equation*}
\Phi (c+(h_{1}(x)))=ca+(h_{2}(x))=1+(h_{2}(x)).
\end{equation*}%
Hence $ca-1\in (h_{2}(x)).$ This gives%
\begin{eqnarray*}
ca-1 &=&l(x)h_{2}(x).\text{ Or} \\
ca-l(x)h_{2}(x) &=&1.
\end{eqnarray*}%
Hence
\begin{equation}
gcrd(a,h_{2}(x))=1.  \label{EQ2}
\end{equation}%
Suppose $lclm[a,h_{2}(x)]=k.$ Then
\begin{equation*}
k=\alpha _{1}a=\alpha _{2}h_{2}\in (h_{2}(x)).
\end{equation*}%
Since $\Phi $ is injective then $\alpha _{1}\in (h_{1}(x)).$ Hence $\alpha
_{1}=t_{1}h_{1}(x)$ and $k=t_{1}h_{1}(x)a.$ But we have\ $%
h_{1}a=r_{2}h_{2}=m.$ Therefore%
\begin{equation}
lclm[a,h_{2}(x)]=h_{1}a.  \label{EQ3}
\end{equation}%
From equations \ref{EQ2} and \ref{EQ3}, we get that $h_{1}(x)$ and $h_{2}(x)$
are (left) similar.

Now suppose $h_{1}(x)$ is (left) similar to $h_{2}(x).$ Then there is $u$
such that
\begin{eqnarray*}
gcrd(u,h_{2}) &=&1\text{ and } \\
lclm[u,h_{2}] &=&h_{1}u
\end{eqnarray*}%
Define%
\begin{equation*}
\Psi :R_{s}/\left( h_{1}(x)\right) \rightarrow R_{s}/\left( h_{2}(x)\right)
\end{equation*}%
by%
\begin{equation*}
\Psi (r+(h_{1}(x))=ru+(h_{2}(x)).
\end{equation*}%
It is clear that $\Psi $ is a module homomorphism. It is left to show that $%
\Psi $ is a bijective function.

Since $gcrd(u,h_{2})=1$ then $c_{1}u+c_{2}h_{2}=1$ for some $c_{1}$ and $%
c_{2}\in R_{s}.$ This implies that%
\begin{equation*}
\Psi (c_{1}+(h_{1}(x))=c_{1}u+(h_{2}(x))=1+(h_{2}(x)).
\end{equation*}%
So for any $r+(h_{2}(x))\in R_{s}/(h_{2}(x))$ we have%
\begin{equation*}
\Psi (rc_{1}+(h_{2}(x)))=r\Psi (c_{1}+(h_{2}(x)))=r+(h_{2}(x)).
\end{equation*}%
Therefore $\Psi $ is surjective. Suppose
\begin{equation*}
\Psi (s+(h_{1}(x))=su+(h_{2}(x))=h_{2}(x)
\end{equation*}
for some $s.$ Then $su\in (h_{2}(x)).$ So,
\begin{equation*}
su=rh_{2}(x)\text{ for some }r.
\end{equation*}%
Since $lclm[u,h_{2}]=h_{1}u$, we have
\begin{equation*}
su=t_{1}h_{1}u.
\end{equation*}
To show $\Psi $ is injective we need to show that $s\in (h_{1}(x)).$ By the
right division algorithm we have%
\begin{equation*}
s=q_{1}h_{1}+r_{1}\text{ where }\deg r_{1}<\deg h_{1}.
\end{equation*}%
This implies that%
\begin{equation*}
su=q_{1}h_{1}u+r_{1}u\Rightarrow r_{1}u\in (h_{2}(x))
\end{equation*}%
Since $lclm[u,h_{2}]=h_{1}u$ then $r_{1}u=t_{2}h_{1}u\in (h_{1}(x)).$ If $%
r_{1}\in (h_{1}(x))$ then
\begin{equation*}
s=q_{1}h_{1}+r_{1}\in (h_{1}(x)),
\end{equation*}
and hence $\Psi $ is injective. If $r_{1}\notin (h_{1}(x))$ then repeat the
right division algorithm again until we get a remainder $r_{i}\in
(h_{1}(x)). $ This implies $r_{i-1},r_{i-2},\ldots ,s\in (h_{1}(x)).$
Therefore $\Psi $ is injective and hence it is an isomorphism.
\end{proof}

\section{Search Results}

The \textbf{Hamming weight enumerator}, $W_C( y)$, of a code $C$ is defined
by
\begin{equation}
W_{C}(y)=\sum_{c\in C}y^{w(c)}=\sum_{i}A_{i}y^{i}  \label{usualh}
\end{equation}%
where $w(c)$ is the number of the nonzero coordinates of the codeword $c$
and $A_{i}=|\{c\in C|w(c)=i\}|,$ i.e. the number of codewords in $C$ whose
weights equal to $i$.

The smallest non-zero exponent of $y$ with a nonzero coefficient in $W_C(y)$
is equal to the minimum distance of the code.

We know that the ring $F[x]$ is a unique factorization domain and the
polynomial $x^{s}-1$ has a unique factorization as a product of irreducible
polynomials in $F[x].$ Things are different in the ring $F[x;\theta ].$ The
skew polynomial ring $F[x;\theta ]$ is not a unique factorization domain and
hence polynomials in general do not have a unique factorization as a product
of irreducible polynomials.

\begin{example}
Consider $x^{4}-1$ over $F=GF(4)$. We have%
\begin{eqnarray*}
x^{2}-1 &=&(x-1)(x-1) \\
&=&(x-a )(x-a ^{2}),
\end{eqnarray*}%
and%
\begin{eqnarray*}
x^{4}-1 &=&(x-1)^{4} \\
&=&(x+a )(x+a ^{2})(x+a )(x+a ^{2}) \\
&=&(x+a )(x+a )(x+a ^{2})(x+a ^{2}) \\
&=&(x+a )(x+a ^{2})(x+1)(x+1).
\end{eqnarray*}
\end{example}

One of the main problems of coding theory is to construct codes with best
possible parameters. There is a well known table of linear codes with best
known parameters over small finite fields \cite{server2}. The computer
algebra system Magma also has such a database \cite{magma}. Researchers
continuously update these tables as new codes are discovered. As the gaps
narrow in the tables, it gets more and more difficult to find new codes.
Many of the new codes discovered in recent years have come from the class of
QC and QT codes (e.g., \cite{cyclicZ4},\cite{optimal},\cite{nine},\cite{ieee}%
). One advantage of studying codes in $F[x;\theta ]$ compared to codes over $%
F[x]$ is that the number of factors of $x^{s}-1$ in $F[x;\theta ]$ is much
larger. Therefore, there are many more skew cyclic and skew QC codes in $%
F[x;\theta ]$ than there are cyclic and QC codes in $F[x].$ This suggests
that it may be possible to find new codes in the ring $F[x;\theta ]$ with
larger Hamming distances. Our search has yielded a number of skew QC codes
with best known parameters. We call such codes ``good codes". Seven of these
codes lead to improvements in the table \cite{server2}. These are called
``new codes". The improvement on minimum distance is 1 unit in each case. We
present these codes in the rest of this section. These results show that the
class of skew QC is a promising class that deserve further attention.

In view of the previos sectiond and the findings obtained therin, our
strategy to search for new codes or good codes is as follows: Choose an
integer $s$, and find a factor $g$ of $x^s-1$ in $F[x;\theta ]$ (where $%
F=GF(4)=\{0,1,a,a^2\}$). Then search for polynomials $f_1,f_2, \dots,f_{l-1}$
so that the skew QC codes of the form $(g,f_1\cdot g, \dots, f_{l-1}g)$ have
large minimum distances. We have used the computer algebra system Magma to
carry out all of the computations.

\begin{example}
We consider a skew $2$-QC code of length $48$. Hence, we need a
factorization of $x^{24}-1$. One such factorization is $x^s-1=g\cdot h$
where $g =x^{12} + ax^9 + x^8 + ax^7 + ax^6 + x^5 + a^2x^4 + ax^3 + ax^2 +
a^2x + a^2$ and $h=x^{12} + ax^9 + x^8 + ax^7 + a^2x^6 + x^5 + ax^4 + ax^3 +
a^2x^2 + a^2x + a$. Letting $f = x^{11} + a^2x^10 + ax^9 + a^2x^7 + x^6 +
a^2x^5 + ax^4 + ax^3 + x + a$, the code generated by $(g,f\cdot g)$ has
parameters $[48,12,24]$ over $GF(4)$. This code has a larger minimum
distance than the previously best known code with the same length and
dimension.

The weight enumerator of this code is as follows:

$W_{C}=1 +3390y^{24}+ 4608y^{25} +19944y^{26} + 25968y^{27} + 99612y^{28}+
124272y^{29}+ 388872y^{30}+427392y^{31}+ 1125315y^{32}+ 958464y^{33}+
2102544y^{34}+1529568y^{35}+ 2798568y^{36}+ 1613664y^{37}+ 2320272y^{38}+
1078272y^{39}+ 1224378y^{40}+ 436608y^{41}+ 345096y^{42}+ 84528y^{43}+
54972y^{44}+ 8112y^{45}+ 2664y^{46}+ 132y^{48}.$
\end{example}

\begin{example}
Let us consider a skew $3$-QC code of length $72.$ We again need a
factorization of $x^{24}-1.$ Here is another factorization of $x^{24}-1$: $%
x^{24}-1=g\cdot h$ , where $g =x^3 + a^2x^2 + 1$ and $h=x^{21} + ax^{20} +
x^{19} + a^2x^{18} + x^{16} + x^{13} + ax^{12} + x^{11} + a^2x^{10} + x^8 +
x^5 +ax^4 + x^3 + a^2x^2 + 1$. Now let $%
f_1=x^{20}+x^{19}+x^{17}+a^2x^{15}+ax^{14}+a^2x^{13}+a^2x^{12}+a^2x^{11}+x^{10} +a^2x^9+x^8+x^7+ax^6+a^2x^5+ax^2+1
$ and $f_2=x^{13} + a^2x^{12} + x^{10} + x^9 + x^8 + a^2x^7 + ax^3 + ax $
and consider the code $C$ generated by $(g,f_1\cdot g, f_2\cdot g)$. It is a
$[72,21,29]$ code and therefore better than the previously best known code
with parameters $[72,21,28]$. The weight enumerator of $C$ is also available
(but not printed here).
\end{example}


In the rest of the examples, we use the trivial factor of 1,
therefore the generators of the codes are of the form $(f_1,f_2,
\dots, f_l)$. We shall refer to such codes as non-degenerate skew
QC codes (since the codes of the form $(f_1g,f_2g, \dots, f_lg)$
with $\deg{g}>0$ are sometimes referred to as degenerate QC codes
in the literature). The polynomials are represented by a list of
coefficients of increasing powers. Hence, the sequence $a001aa^21$
represents the polynomial $x^6+a^2x^5+ax^4+x^3+a$.

\begin{example}
A $[48,16,20]$-code generated by $f1=0a^2a^2a0a^210a^20a11a^2a^21$, $%
f2=100a^20a^2a^2aa^21a^21a^20a^20$, $f3=a^2aa0a^20aa1a^2aaa0aa$.
\end{example}

\begin{example}
A $[96,16,49]$-code generated by $f1=0a^2a^21aa^20aa100a^2a0a$, $%
f2=1a^2a^2aa00a^2a^211a^21a0a^2$, $f3=0a^2a^200aaaa^21a1a^20aa^2$, $%
f4=a0a^200a0a^2aa0aa1a^21$, $f5=a^2011011a^21a1a^2a111$, $%
f6=a100a^2a^2a^2a1a001aa^2a^2$
\end{example}

\begin{example}
A $[100,20,47]$-code generated by\newline
$f1=a00a^2a^2001a^2a^2a^2011a1a^2a11$,\newline
$f2=01a^20a1a01a^21a1a01001a^2$,\newline
$f3=a1aa1001aa^20000a^2a1a^2a^21$,\newline
$f4=1a1aa11a^2a^2aa^20a^2a0010a^21$,\newline
$f5=a^20111aa^21a^2aa^2a^2a0a^201a11$
\end{example}

\begin{example}
A $[140,20,72]$-code generated by\newline
$f1=1a^2a^2aa1a10aa^210a01a^2a^201$,\newline
$f2=aa0a^201a^2aa0a0a1aa1a10$,\newline
$f3=a^2a^2a^21aa^2a1a0aaa^2a^20aa0aa$,\newline
$f4=10001aaa^20a010a^2a^2a0010$,\newline
$f5=a11001a1a^2a^21aa^210aa^21a^2a$,\newline
$f6=a^20a^210a^211a^2a^2a^21a^2a^20a^20110$,\newline
$f7=a^21011000a^2a^201a^201a^2aa^2a^21$
\end{example}

\begin{example}
A $[110,22,51]$-code generated by\newline
$f1=1a^2010aa0a^201a^2100a0a^2a0a^20$,\newline
$f2=a^2a0101aa^21a^2a^21a^211aaa^200a^21$,\newline
$f3=00a^2a00a^201a^2aa100a0a^2a11a^2$,\newline
$f4=01a01010a^211a01100a^2a^2a1a$,\newline
$f5=a^20a0a^2a^2a00a^2a10a0aaa1a^21a$
\end{example}

We summarize the rest of the results of our search that yielded good codes
in the following three tables.


\begin{table}[tbp]
\caption{Parameters and Generators of the Good skew QC Codes of index 2 }
\label{tab:cyclic}
\begin{center}
\vspace*{0.2in}
\begin{tabular}[t]{|p{1.4 cm}|p{4.5 cm}|p{4.5 cm}|}
\hline
\multicolumn{1}{|c|}{Parameters} & \multicolumn{1}{c|}{$g$} &
\multicolumn{1}{c|}{$f$} \\ \hline\hline
$[40,9,21]$ & ${\small {aa^200a1a^2a^210a1}}$ & ${\small {a0000aa^201}}$ \\
\hline
$[40,10,20]$ & $a^2a^2a01a0aa^211$ & $a^2a^2a^200a1aa^21$ \\ \hline
$[40,11,19]$ & $a10aaaa^21a1$ & $a11aa^2100a^201$ \\ \hline
$[40,12,18]$ & $101a^200aa^21$ & $1a100aaaaa^2a1$ \\ \hline
$[40,14,16]$ & $10a^21a01$ & $11aaa^2a011$ \\ \hline
$[40,16,15]$ & $10001$ & $aa0a^201011a^21a^20a^2$ \\ \hline
$[40,17,14]$ & $a^21a^21$ & ${\tiny {a^2a1a^210a^20a0a^2a^2a^20a^211}}$ \\
\hline
$[44,12,20]$ & $11111aa11a^21$ & $a000a^2a^2aa0a^2a1$ \\ \hline
$[48,11,24]$ & $1aa^21a0a^2a100101$ & $aaaaa^21a10a1$ \\ \hline
$[48,12,23]$ & $a110a^21a^2a11a^2a^21$ & $1a^20110010a11$ \\ \hline
$[48,13,22]$ & $a^2aa^200a100111$ & $0001a0a^2a^2a^2aaa^21$ \\ \hline
$[48,14,21]$ & $11a^2a^2a10a^2101$ & $1a11010a1a^2aa11$ \\ \hline
$[48,15,20]$ & $a^2a01a1a^2111$ & $1aa^2a^2a^2a^2a^2a0aa1a11$ \\ \hline
$[48,16,19]$ & $a00101a^2a1$ & $a^2a0a1a1aa1aa^2aa^2a^21$ \\ \hline
$[52,13,24]$ & $a^2a^2aa^210a^21a11aa1$ & $aa^211aa^21a^2a^2a^2a^211$ \\
\hline
$[60,11,32]$ & ${\small {aaa^2a^2a^2a^211aa110000aa11}}$ & ${\small {%
aa^2aa^20aa^2a1a^21}}$ \\ \hline
$[60,14,28]$ & $11a^2a11a^2a^2a^2a^2aa^2a^2a^2a^2a1$ & $a^20a^2a0a110a^21011$
\\ \hline
\end{tabular}%
\end{center}
\end{table}

\section{Conclusion}

In this paper, we study the structure of 1-generator skew QC codes in the
non-commutative ring $F[x;\theta ]$. We have shown that skew QC codes are
left submodules of the ring $R_{s}^{l}=\left( F[x;\theta ]/(x^{s}-1\right)
)^{l}.$ We also introduced the notion of similar polynomials in the ring $%
F[x;\theta ]$ and showed that parity-check polynomials are unique up to
similarity. Our search results showed the construction of several new linear
codes with Hamming distance larger than the Hamming distance of the best
linear codes with similar parameters. An important problem that needs to be
addressed is an efficient method of obtaining all factorizations of $x^{n}-1$
in the skew polynomial ring. Also, a BCH type bound for skew cyclic and skew
QC codes is a future topic of interest.



\begin{table}[tbp]
\caption{Parameters and Generators of the Good skew QC Codes of index 3 and
4 }
\label{tab:cyclic}
\begin{center}
\vspace*{0.2in}
\begin{tabular}[t]{|p{1.4 cm}|p{2 cm}|p{4.3 cm}|p{4.3 cm}|p{4.5 cm}|}
\hline
\multicolumn{1}{|c|}{Parameters} & \multicolumn{1}{c|}{$g$} &
\multicolumn{1}{c|}{$f_1$} & \multicolumn{1}{c|}{$f_2$} &
\multicolumn{1}{c|}{$f_3$} \\ \hline\hline
$[48,11,24]$ & $1a01a^21$ & $1aa^21a1aa111$ & $a^2a^2aa111$ & - \\ \hline
$[48,13,22]$ & $a1a1$ & $aa^2a0aa01aa101$ & $aa^2a^2011a11$ & - \\ \hline
$[48,14,21]$ & $a^2a^21$ & $a^2a^21aaaa^2a^20aaa01$ & $1a^2aa01a^21$ & - \\
\hline
$[48,15,20]$ & $a1$ & $0a^2a^211a0a^21aa^20aa1$ & $aaa1100a^21$ & - \\ \hline
$[54,13,26]$ & $a1a1a1$ & $a1a^201a^2a0a^2a101$ & $a^2a^21a^20001$ & - \\
\hline
$[54,15,24]$ & $aa11$ & $a^2a^2a^21aaa1aa^20a^20a1$ & $a0aa^20a^2a^2a11a^21$
& - \\ \hline
$[60,14,28]$ & $a^20aaa01$ & $111a00aa^210a0a^21$ & $10a111aa^2a1a^2aa1$ & -
\\ \hline
$[60,18,25]$ & $aa1$ & $1a1a01aa^2a^2a^2aa^200a^21a^21$ & $%
a^2a^2a011a^21aa^2101$ & - \\ \hline
$[60,19,24]$ & $a1$ & $00a^2a^2aaaa^21a^21a^2a^2a^2aa^2111$ & $%
1aa^20a^21a11a^2a11$ & - \\ \hline
$[72,21,28]$ & $1aa1$ & $000a1aa^21a^200aaaa^2aa^2a^2011$ & $%
0aa1a^20a^211aa00a1$ & - \\ \hline
$[72,19,30]$ & $a^2a1001$ & $1a^2a10a^2a0a^2aa01aaaaa1$ & $1a^21a^2a^2a1111$
& - \\ \hline
$[72,15,34]$ & $aaa^21a1aa^211$ & $a^20a0aa^2a^211a^2aa0a^21$ & $%
01aa10a0a^2a^21$ & - \\ \hline\hline
$[56,11,29]$ & $11a1$ & $0a^200a01a^2$ & $0a^21a1aa^21$ & $1a^210aa1a^2a^2$
\\ \hline
$[56,12,28]$ & $101$ & $aa^20a^2a100aa$ & $a^21a^2a^2a0aa^2aa^2a^2$ & $%
a^2a0aaa^21$ \\ \hline
$[64,13,32]$ & $a1a1$ & $11aa1011a^2a^2a$ & $a^2000a$ & $10a1a^2011$ \\
\hline
$[64,14,31]$ & $101$ & $0a1a01a^2a^20aa^2$ & $1aaa^2a000a^2a^21a$ & $%
aa^2a^211a^20aa$ \\ \hline
$[64,15,30]$ & $a^21$ & $aa11aa^21aaaa$ & $a^2a^201aa^2aa^2001$ & $%
1a^201a1a1 $ \\ \hline
$[80,17,38]$ & $1111$ & $a^2a^201aaaa01a^2aa^2a$ & $11a11aa00aaa^2$ & $%
a^2111aa^21aaa1a^20a^2$ \\ \hline
$[72,15,34]$ & $1a^2a1$ & $0a1011a^2a^20aa^21a$ & $a010a0a^2a^210aa$ & $%
a^201aaa^2$ \\ \hline
$[96,20,44]$ & $10101$ & $11a^2a000110a^2a^2a0aa^2aa^2$ & $%
a0111aaa^20aaa^20a^2a^2$ & $a^2a^210a^210a$ \\ \hline
$[96,23,41]$ & $a^21$ & $01aa^2a^21a^200aa^210aaa^2a1$ & $%
0aaa^2a^2a^2a^2aa0a^2a00a^2a^2a$ & $a^2101a^2aaa^21a01a10aa$ \\ \hline
$[112,24,48]$ & $1a^2a01$ & $00a1a^2a^21000a^2a^21a^200a0a101$ & $%
a^2aa^2a^2a^21aa^2a01a1a^2a^211a1$ & $00a^2a00a^2a0a^21a^2101$ \\ \hline
$[112,22,50]$ & $aaa^2a^2a11$ & $a^21a^2a1a0a^2a^21110a^211$ & $%
11a1a^2a^210a11a0a^2$ & $0a^20aaa^2a100a^21a^2a^2$ \\ \hline
$[120,21,57]$ & $aa1aa^2a^211a1$ & $a^21a^2aa^2a011aa00aa1$ & $%
00aa^2aa^21a^20a^200a^2a^20a^2a^2$ & $1aa0a^2a^2100a^2$ \\ \hline
$[120,23,54]$ & $a010a0a^21$ & $1a^211a0a00aa^2100a$ & $%
010a110aa^200a^211a^2aa^2$ & $aa0aa^200a^21a11001a1a$ \\ \hline
$[120,25,52]$ & $11a^2a^211$ & $001a00a^2a^21a^21aaaa111$ & $%
a^20a^2aa1a^21a^2aaa11a^2a$ & {\small {$%
a^21a^2a00aaa^2a^21a000aaa^2a^2a^21a^2$}} \\ \hline
\end{tabular}%
\end{center}
\end{table}

\newpage

\newpage

\begin{table}[tbp]
\caption{Parameters and Generators of the Good non-degenerate skew
QC Codes of index up to 4 } \label{tab:cyclic}
\begin{center}
\vspace*{0.2in}
\begin{tabular}[t]{|p{1 cm}|p{4 cm}|p{4 cm}|p{4 cm}|p{4 cm}|}
\hline
\multicolumn{1}{|c|}{{\small {Parameters}}} & \multicolumn{1}{c|}{$f_1$} &
\multicolumn{1}{c|}{$f_2$} & \multicolumn{1}{c|}{$f_3$} &
\multicolumn{1}{c|}{$f_4$} \\ \hline\hline
${\small {[40,20,12]}}$ & {\small {$a^21aaa^2a^2aaaa1aa1010aaa$}} & {\small {%
$0a1 0a1aa11a10a 0 1 1 0 1 0$}} & $-$ & - \\ \hline
${\small {[30,10,14]}}$ & {\small {$a^2aa00a10aa$}} & {\small {$%
000a^2a^2a1a1a^2$}} & {\small {$0a^21aa^20aa^21a$}} & - \\ \hline
${\small {[36,12,16]}}$ & {\small {$0 a a^2 0 0 a 0 0 a^2 a^2 a^2 a$}} &
{\small {$011a1a^21a^20aa^20$}} & {\small {$a^2 a 0 0 a^2 a a a 1 1 a a$}} &
- \\ \hline
${\small {[42,14,18]}}$ & {\small {$a a a a^2 1 a^2 a a 1 0 a^2 a^2 a a$}} &
{\small {$a^2 a^2 0 a^2 a^2 a a a 1 0 0 a^2 0 0$}} & {\small {$1 0 0 a^2 a^2
1 a^2 a^2 a^2 1 a^2 1 0 0$}} & - \\ \hline
${\small {[48,16,19]}}$ & {\small {$a^2 a^2 0 0 a^2 a a a 0 1 1 0 a a 1 a$}}
& {\small {$0 0 1 0 a^2 0 a^2 a 1 a 0 a a a^2 a a$}} & {\small {$a^2 1 a^2 1
a^2 1 0 1 a 1 a^2 a 0 0 a^2 0$}} & - \\ \hline
${\small {[66,22,25]}}$ & {\small {$a^2 1 0 a a^2 a^2 a 1 a 0 0 a^2 0 0 0 1
0 a 1 a 0 1$}} & {\small {$0 0 1 a 0 a^2 0 1 a^2 0 a^2 a a a 0 a 0 1 1 a a 0$%
}} & {\small {$1 0 0 0 0 0 a^2 1 1 0 a a a a a^2 0 a 0 a 0 a^2 a^2$}} & - \\
\hline
${\small {[40,10,20]}}$ & {\small {$1 a^2 a^2 a a^2 a^2 a^2 a^2 a^2 a$}} &
{\small {$0 a^2 1 1 0 a a^2 a 1 1$}} & {\small {$a^2 a a^2 a 0 0 0 a^2 a^2 1$%
}} & {\small {$a^2 a 0 a^2 1 a 0 a a 0$}} \\ \hline
${\small {[48,12,23]}}$ & {\small {$0 1 a^2 a^2 a a^2 0 a^2 a^2 a^2 0 1 $}}
& {\small {$a^2 0 a a a^2 1 a a a^2 0 a 1$}} & {\small {$a a a a^2 a^2 1 0
a^2 a^2 0 0 0$}} & {\small {$0 a^2 a a a a a 0 1 a^2 a 0$}} \\ \hline
${\small {[72,18,32]}}$ & {\small {$a 0 a^2 1 1 a a^2 a^2 0 a 1 1 a^2 a^2 0
a^2 0 0 $}} & {\small {$1 0 0 1 0 1 1 a 1 a 0 a a 0 1 1 a^2 0$}} & {\small {$%
0 1 0 a a^2 0 a^2 a 0 a^2 a^2 a 0 0 0 0 a 0$}} & {\small {$a a 1 0 0 0 a 1 0
a 0 0 1 a 0 a^2 1 a$}} \\ \hline
${\small {[80,20,35]}}$ & {\small {$a^2 a^2 1 1 a 1 a^2 a 1 1 a a 0 1 a^2 0
0 a a^2 0$}} & {\small {$1 0 a^2 a a^2 a a 1 a a^2 a 0 a^2 1 a a a^2 0 a^2 0$%
}} & {\small {$a^2 1 a a^2 0 a^2 a^2 1 1 1 a^2 0 1 0 1 1 a a 1 a$}} &
{\small {$1 0 a 0 a^2 0 0 a^2 a^2 a^2 1 1 0 a^2 a^2 1 1 a a 0$}} \\ \hline
${\small {[96,24,40]}}$ & {\tiny {$a 0 a 0 0 1 0 1 1 a^2 a^2 a 1 1 1 1 a^2
a^2 0 a 0 0 0 a^2$}} & {\tiny {$0 a^2 1 1 1 0 0 0 a^2 1 0 0 0 1 0 1 a a a^2
0 a^2 0 1 a$}} & {\tiny {$a 1 a a^2 0 a^2 a^2 a a^2 a 1 a^2 a^2 0 0 1 a^2
a^2 a a 0 a 0 a$}} & {\tiny {$a a 0 1 a 1 0 1 a^2 1 a 0 1 0 a^2 0 1 1 a 0 0
a a a$}} \\ \hline
\end{tabular}%
\end{center}
\end{table}

\newpage

\newpage

\begin{table}[tbp]
\caption{Parameters and Generators of the Good non-degenerate skew
QC Codes of larger indices } \label{tab:cyclic}
\begin{center}
\vspace*{0.2in}
\begin{tabular}[t]{|p{1.2 cm}|p{1.2 cm}|p{1.2 cm}|p{1.2 cm}|p{1.2 cm}|p{1.2 cm}|p{1.2 cm}|p{1.2 cm}|p{1.2 cm}|p{1.2 cm}|p{1.2 cm}|}
\hline
\multicolumn{1}{|c|}{{\small {[60,12,31}}} & \multicolumn{1}{c|}{{\small {%
[60,10,33]}}} & \multicolumn{1}{c|}{{\small {[100,20,46]}}} &
\multicolumn{1}{c|}{{\small {[110,22,50]}}} & \multicolumn{1}{c|}{{\small {%
[72,12,38]}}} & \multicolumn{1}{c|}{{\small {[96,16,48]}}} &
\multicolumn{1}{c|}{{\small {[70,10,40]}}} & \multicolumn{1}{c|}{{\small {%
[140,20,71]}}} & \multicolumn{1}{c|}{{\small {[96,12,54]}}} &
\multicolumn{1}{c|}{{\small {[160,20,84]}}} & \multicolumn{1}{c|}{{\small {%
[144,16,80]}}} \\ \hline\hline
{\tiny {$a^2 a^2aa^2a^2a$} {$1a1010$}} & {\tiny {$1 a^2 0 1 0a^2$} {$0 a^2
10 $}} & {\tiny {$a^2 a^2 0 1 a a1 a^2$} {$00 1a a^2 a^20 a 1 $} {$0 1 a^21
a^2 0$}} & {\tiny {$1 0 a a0$} {$a^2 0 1a^2 1 0 0$} {$a 0 a^2 a 0 a^2 0$}} &
{\tiny {$0 a^2 0 a^2 0 1$} {$a a^2 a^2 1 a 0$}} & {\tiny {$00a a^2 a a$} {$%
0a^20a^2 0 a^2$} {$a^2 0a^2 a^2$}} & {\tiny {$0 a^2 a a^2 0 $} {$0 1 1 a^2 a$%
}} & {\tiny {$1 a^2 1 a^2 a a^2 $} {$a^2 a 1a^2a^2 0 a^2$} {$a a^2 0 a^2 a
a^2 0$}} & {\tiny {$a^2 a^20 a a^20$} {$1 0 0 0 a 1$}} & {\tiny {$a a^2 a^2
a^2 1 0 $} {$a^2 0 0a 1 a 1$} {$a a 1 a^2 0 a a^2$}} & {\tiny {$1 1 a 1 0
a^2 1$} {$a^2 0 0 0 1 0 a a^2 a$}} \\ \hline
{\tiny {$a a 1 a^2 0 1 1 $} {$a a^2 1 a a^2$}} & {\tiny {$a a 1 a^2 1 a$} {$%
0 0 a^2 0$}} & {\tiny {$1 1 a^2 a a^2 0 a$} {$1 1 a^2 0 1 0 a$} {$a^2 a^2
a^2 a 1 0$}} & {\tiny {$a^2 a 0 1 0 1 a$} {$a^2 1 a^2 a^2 1 a^21 1 $} {$aa
a^2 0 0 a^2 1$}} & {\tiny {$a 0 a^2 a^2 a 1$} {$a^2 1 1 a 1 a^2$}} & {\tiny {%
$a a a 0 1 0 1$} {$1 1 1 1 a^2 1 0$} {$a^2 a^2$}} & {\tiny {$a a 0 1 0 0 $} {%
$1 0 1 a^2$}} & {\tiny {$a^2 0 1 1 a a^2 a$} {$0 a^2 a^2 1 a^2 a 0$} {$a 0 1
a^2 a a$}} & {\tiny {$0 a^2 0 1 0 a^2$} {$a^2 a a a a^2 0$}} & {\tiny {$0 a
a 1 a^2 0 0$} {$1 a^2 1 a a^2 a a$} {$a^2 a 0 a^2 a a$}} & {\tiny {$a^2 a^2
0 a a a^2$} {$a a^2 a^2 a^2 0 a$} {$1 1 a^2 1$}} \\ \hline
{\tiny {$0 a^2 1 a a^2 a^2$} {$a a a a^2 1 0$}} & {\tiny {$a 0 a 1 1 0$} {$%
a^2 a^2 a 1$}} & {\tiny {$1 a a 0 0 1 a$} {$a^2 1 0 a^2 a^2 a 0$} {$1 a a^2
1 1 0$}} & {\tiny {$0 a a^2 0 1 a 0 a$} {$a^2 a a a a^2 1 a^2$} {$1 0 1 a^2
a^2 a^2 1$}} & {\tiny {$a^2 1 a a^2 0 1$} {$0 0 a a^2 0 a^2$}} & {\tiny {$0
1 1 1 a^2 0 0$} {$aa^2 a a^2 a^2 a^2$} {$0 a a$}} & {\tiny {$a^2 a^2 1 1 a^2
$} {$a a^2 1 a^2 0$}} & {\tiny {$a a a^2 a a^2 a a^2$} {$1 a 1 a^2 a 1 0$} {$%
a a^2 a a a^2 0$}} & {\tiny {$1 a a 0 1 a$} {$1 0 1 a a^2 a^2$}} & {\tiny {$%
a^2 1 0 0 a^2 a 1 a$} {$a 1 0 a^2 a^2 0 a^2$} {$0 a a a^2 1$}} & {\tiny {$0
1 a^2 1 1 a a$} {$1 a 0 1 a a^2$} {$1 a^2 a$}} \\ \hline
{\tiny {$a^2 0 a^2 a 1 a^2$} {$a 0 a^2 0 1 1$}} & {\tiny {$0 0 1 a a$} {$1 a
a 0 a$}} & {\tiny {$1 a 0 0 a^2 a a^2$} {$1 0 0 a^2 a^2 a 0$} {$1 0 1 1 a^2
a $}} & {\tiny {$1 a 1 a 0 0 0 0$} {$a a^2 a^2 a 0 a^2 1$} {$1 0 1 a 1 a^2 0$%
}} & {\tiny {$1 1 a a^2 0 0$} {$1 a^2 a 1 a a$}} & {\tiny {$0 a 1 a 0 a^2 1$}
{$a 0 0 1 0 1$} {$a^2 1 1$}} & {\tiny {$0 1 a^2 0 1 $} {$a^2 0 1 1 a^2$}} &
{\tiny {$a^2 a a^2 0 a^2 a a^2$} {$1 a^2 a a 0 1 0$} {$a a^2 a a 0 1$}} &
{\tiny {$a^2 a a^2 0 a 0 $} {$a a^2 0 1 a^2 a^2$}} & {\tiny {$0 1 a^2 1 1 a
1 $} {$1 a a 1 a^2 0 1$} {$1 a^2 0 a a a$}} & {\tiny {$1 1 0 a^2 0 a a$} {$0
0 a a a^2 a$} {$a a^2 a^2$}} \\ \hline
{\tiny {$a 1 0 0 a^2 1$} {$1 a a^2 1 a a$}} & {\tiny {$1 1 a 0 1 a^2$} {$a^2
1 a^2 0$}} & {\tiny {$1 a^2 1 a^2 1 a 0$} {$1 0 a^2 0 1 1 1$} {$a^2 0 a^2 1
a 1$}} & {\tiny {$a^2 1 0 0 0 0 a^2 0$} {$a^2 a a a 0 1 a^2$} {$a a^2 a a
a^2 a a^2$}} & {\tiny {$1 a^2 a a^2 a 1$} {$a a^2 a^2 a 1 1$}} & {\tiny {$0
0 a 1 0 a^2 $} {$1a 1 a 1 a$} {$a 1 0 a^2$}} & {\tiny {$0 0 a^2 0 a $} {$1
a^2 a a^2 a^2$}} & {\tiny {$1 a^2 0 1 0 a a^2$} {$1 a a^2 a 0 a^2 a^2$} {$a
0 a 0 a^2 1$}} & {\tiny {$1 1 0 1 0 1$} {$a 0 a^2 a 0 1$}} & {\tiny {$0 a 1
0 1 a 0 1$} {$0 a^2 a^2 a^2 a 1$} {$a^2 0 0 1 0 0$}} & {\tiny {$a 0 1 0 a 0
111$} {$a a^2 0 a^2 a a^2 0$}} \\ \hline
- & {\tiny {$a^2 1 0 a 1$} {$0 a^2 a 1 0$}} & - & - & {\tiny {$a^2 a^2 a^2 a
a 1$} {$a^2 a^2 1 0 1 a^2$}} & {\tiny {$a a^2 1 a^2 0 1 $} {$a^2 a^2 a^2 a^2
a$} {$0 1 1 1 1$}} & {\tiny {$a^2 a 0 0 1 1 1 a 0 a$}} & {\tiny {$a^2 1 a^2
0 1 a^2 a$} {$1 0 a 1 0 0 a^2$} {$a^2 a a 1 0 a^2$}} & {\tiny {$0 a a^2 1 1
0 $} {$a^2 0 0 0 a a$}} & {\tiny {$1 a a 0 0 a^2 1$} {$a^2 1 0 1 a 0 1$} {$0
1 a^2 0 0 1$}} & {\tiny {$0 a 0 a^2 1 0 a^2 1 $} {$0 1 0 1 1 1 1 a^2$}} \\
\hline
- & - & - & - & - & - & {\tiny {$0 a^2 1 0 a^2$} {$a a^2 1 a^2 a$}} & {\tiny
{$0 1 a^2 a^2 1 a^2 a$} {$1 a^2 a^2 a^2 a^2 1$} {$a^2 0 a^2 0 0 0 a^2$}} &
{\tiny {$a^2 a a a 0 a^2$} {$1 0 a^2 0 a^2 a^2$}} & {\tiny {$1 0 1 1 a a^2
a^2 1$} {$a^2 a a a a^2 1 0$} {$a 1 1 a^2 a^2$}} & {\tiny {$a a 1 1 a 0 1
a^2 $} {\ $1 1 a 1 a^2 a 1 a$}} \\ \hline
- & - & - & - & - & - & - & - & {\tiny {$a 0 0 1 a^2 a^2$} {$a^2 a 1 0 a 0$}}
& {\tiny {$1 a^2 a 0 a a 1 0$} {$a^2 0 a a 1 1 1 a$} {$1 a^2 0 a$}} & {\tiny
{$a 1 0 a^2 1 0 a 0$} {$a^2 a^2 1 1 0 0 a a$}} \\ \hline
- & - & - & - & - & - & - & - & - & - & {\tiny {$a^2 0 1 a 1 a a^2 1$} {$0 0
a a a 0 a^2 a$}} \\ \hline
\end{tabular}
\end{center}
\end{table}

\end{document}